\documentclass[journal]{vgtc}                %
\ifpdf%
  \pdfoutput=1\relax                   %
  \usepackage{graphicx}                %
  \DeclareGraphicsExtensions{.pdf,.png,.jpg,.jpeg} %
\else%
  \usepackage{graphicx}                %
  \DeclareGraphicsExtensions{.eps}     %
\fi%

\graphicspath{{figures/}{pictures/}{images/}{./}} %

\usepackage{microtype}                 %
\PassOptionsToPackage{warn}{textcomp}  %
\usepackage{textcomp}                  %
\usepackage{mathptmx}                  %
\usepackage{times}                     %
\usepackage{cite}                      %
\usepackage{tabu}                      %
\usepackage{booktabs}                  %

\usepackage{algorithm}
\usepackage[noend]{algpseudocode}

\onlineid{1233}

\vgtccategory{Data Transformations}

\usepackage{amssymb,amsmath,amsthm,bbm}
\usepackage{wasysym}
\usepackage{mathrsfs}

\newcommand{\R}{\mathbb{R}}

\newcommand{\X}{{\mathbb X}}

\newcommand{\BB}{\mathcal{B}}

\newcommand{\MM}{\mathcal{M}}

\newcommand{\Dgm}{\mathrm{Dgm}}

\newcommand{\e}{\varepsilon}
\renewcommand{\phi}{\varphi}

\newcommand{\dLimit}{\ensuremath{\mathsf{L}}}
\newcommand{\spread}{\ensuremath{\mathsf{S}}}

\newcommand{\dgmf}{\ensuremath{\mathrm{Dgm}({\mathcal{M}_f)}}}
\newcommand{\dgmg}{\ensuremath{\mathrm{Dgm}({\mathcal{M}_g)}}}

\newtheorem{theorem}[equation]{Theorem}
\newtheorem{corollary}[equation]{Corollary}

\newtheorem{proposition}[equation]{Proposition}
\newtheorem{definition}[equation]{Definition}

\title{Computing a Stable Distance on Merge Trees}

\author{Brian Bollen, Pasindu Tennakoon, and Joshua A. Levine}
\authorfooter{
\item
 Brian Bollen is with the Department of Mathematics at The University of Arizona. E-mail: bbollen23@math.arizona.edu.
\item
 Pasindu Tennakoon is with the Department of Computer Science at The University of Arizona. E-mail: pasindut@cs.arizona.edu.
\item
 Joshua A. Levine is with the Department of Computer Science at The University of Arizona. E-mail: josh@cs.arizona.edu.
}

\shortauthortitle{Bollen \MakeLowercase{\textit{et al.}}: Computing a Stable Distance on Merge Trees}
\abstract{
Distances on merge trees facilitate visual comparison of collections of scalar fields. Two desirable properties for these distances to exhibit are 1) the ability to discern between scalar fields which other, less complex topological summaries cannot and 2) to still be robust to perturbations in the dataset. The combination of these two properties, known respectively as stability and discriminativity, has led to theoretical distances which are either thought to be or shown to be computationally complex and thus their implementations have been scarce. In order to design similarity measures on merge trees which are computationally feasible for more complex merge trees, many researchers have elected to loosen the restrictions on at least one of these two properties. The question still remains, however, if there are practical situations where trading these desirable properties is necessary. Here we construct a distance between merge trees which is designed to retain both discriminativity and stability. While our approach can be expensive for large merge trees, we illustrate its use in a setting where the number of nodes is small. This setting can be made more practical since we also provide a proof that persistence simplification increases the outputted distance by at most half of the simplified value. We demonstrate our distance measure on applications in shape comparison and on detection of periodicity in the von Kármán vortex street.
} %


\CCScatlist{ %
 \CCScat{K.6.1}{Management of Computing and Information Systems}%
{Project and People Management}{Life Cycle};
 \CCScat{K.7.m}{The Computing Profession}{Miscellaneous}{Ethics}
}

\teaser{
  \centering
  \includegraphics[width=0.98\linewidth]{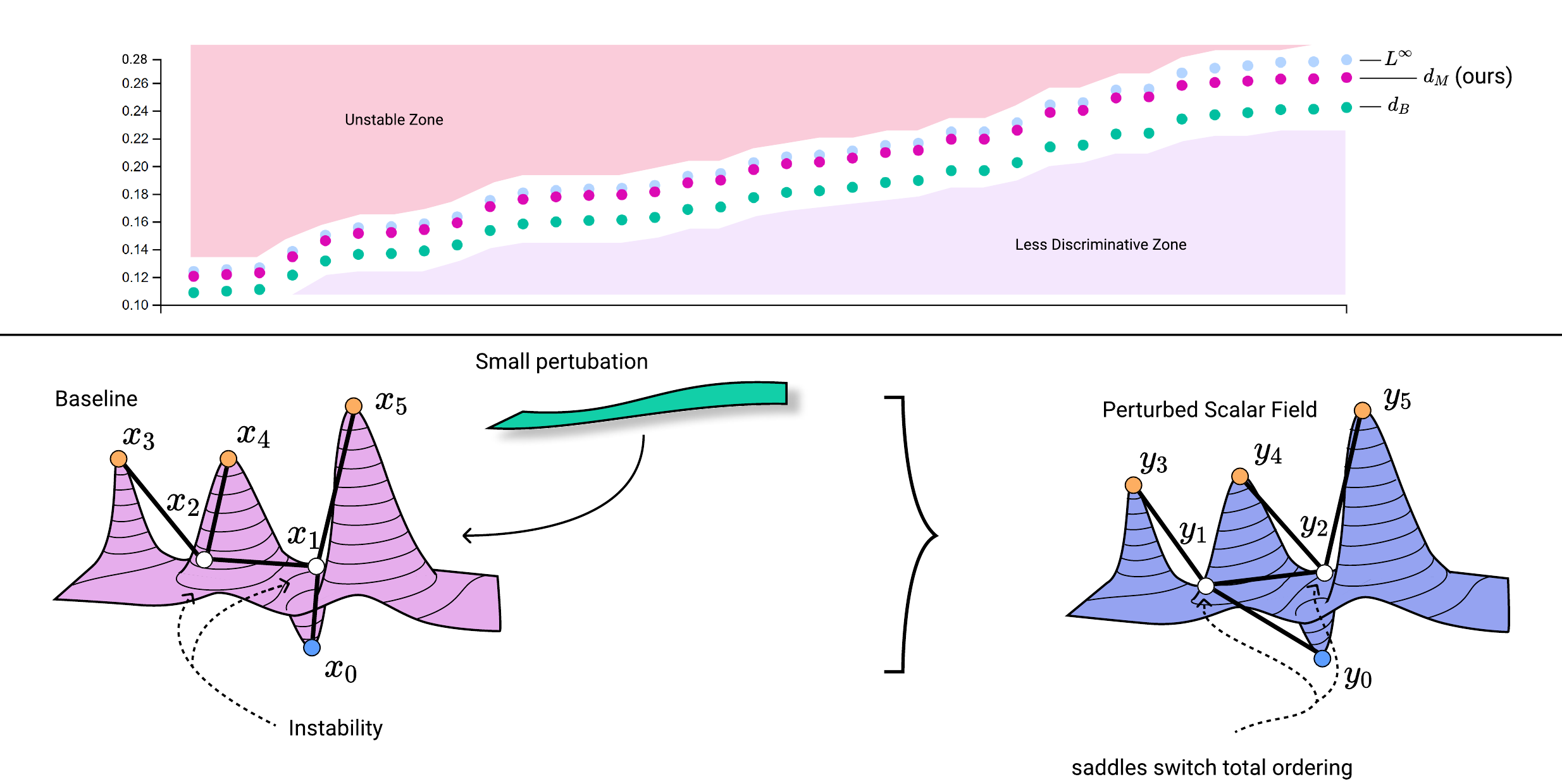}
  \caption{Starting with a baseline scalar field and merge tree, a small perturbation may change the topology of a merge tree. The top graph plots the bottleneck, merge tree matching, and $L^{\infty}$ distance between the baseline and 36 perturbed scalar fields which all exhibit a horizontal instability. The merge tree matching distance is shown here to lie between the bottleneck distance and the $L^{\infty}$ distance even when faced with these instabilities.}
  \label{fig:teaser}
}

\vgtcinsertpkg

\begin{document}

\firstsection{Introduction}
\maketitle
\label{sec:introduction}

Topological descriptors in topological data analysis (TDA) have been used extensively to identify and summarize features of interest in scalar fields in a wide variety of domains such as nuclear energy~\cite{duke2012visualizing}, turbulent mixing~\cite{laney_vis06}, shape analysis~\cite{Hilaga2001}, porous materials~\cite{Gyulassy2007Topologically}, combustion~\cite{bremer2010interactive}, and chemistry~\cite{gunther2014characterizing}. Most topological descriptors fall into one of three categories: the set-based descriptors, such as the persistence diagram \cite{EdelsbrunnerLZ02,Steiner2009}; the graph-based descriptors such as the Reeb graph~\cite{reeb46,pascucci07,tierny_vis09}, contour tree~\cite{Carr2003}, and merge tree; and the complex-based descriptor such as the Morse-Smale complex \cite{edelsbrunner03b,gyulassy_vis08}.

In a visualization setting, we are often posed with the question of how similar two datasets are to one another. Since these topological descriptors have been used for individual analysis of features in the dataset, we can use measures of similarity between the topological descriptors to produce a similarity measure between the underlying datasets~\cite{Yan2021}. For instance, on persistence diagrams, distances such as the bottleneck distance and Wasserstein distance have been used effectively to gauge this similarity. Graph-based structures have also seen a wide variety of distances such as the interleaving \cite{Morozov2013,deSilva2016}, functional distortion \cite{Bauer2014}, universal \cite{Bauer2020}, and multiple edit distances \cite{DiFabio2012,DiFabio2016, Bauer2016}. These graph-based and set-set based descriptor distances have all been proven to be \textbf{stable} -- a property which indicates that the distance is robust to perturbations of the underlying dataset. Furthermore, the graph-based distances have all been shown to be more \textbf{discriminative} than the bottleneck distance, i.e. they can discern between differences in the datasets which the bottleneck distance may not. The combination of retaining stability and discriminativity to the bottleneck distance makes graph-based distances desirable.

However, these theoretical distances all have related problems which imply that these distances are computationally complex. For example, the functional distortion distance is a version of the Gromov-Hausdorff distance \cite{Gromov1981} which is known to be $\mathsf{NP}$-hard to approximate within a factor of 3 \cite{Agarwal2018}. Similarly, the Reeb graph edit distance is heavily related to the graph edit distance (GED) which is known to be $\mathsf{NP}$-hard \cite{garey1979computers,Zhiping2009} and determining if two Reeb graphs have an interleaving distance of $\e$ is known to be $\mathsf{NP}$ \cite{deSilva2016}.

In order to construct similarity measures which are computationally feasible, researchers have constructed new distances on the graph-based descriptors (most notably the merge tree) which have ultimately loosened the restrictions on stability while attempting to instead only maintain discriminativity to the bottleneck distance \cite{Saikia2014, Sridharamurthy2018, Beketayev2014}. These distances sacrifice theoretical properties for computational feasibility. In this work, we ask if it is possible to construct a distance on merge trees that is still practical to use despite bounded computational complexity. We couple this with an approximation bound on this distance based on persistence simplification.

\subsection{Contributions}
We choose to focus our efforts on constructing a distance which can experimentally be shown to retain stability \textit{and} discriminativity. To the best of our knowledge, this is the first distance with an implementation which is shown to be both discriminative and stable. Instead of designing a distance for all graph-based descriptors, we follow the lead of several other experimental distances \cite{Sridharamurthy2018,Saikia2017,Beketayev2014} and focus on the simplest -- the merge tree. Our algorithm matches the features of one merge tree to another and computes a cost of this matching which is heavily inspired by the universal distance \cite{Bauer2021}. 

More specifically, this work will contribute the following:
\begin{itemize}
    \item Define an extended semipseudometric on merge trees by first encoding the features of the merge tree using branch decomposition trees;
    \item Construct an algorithm for this distance which utilizes the A*-search algorithm to find a matching between vertices of two branch decomposition trees;
    \item Prove that persistence simplification of the dataset increases our distance by at most half the simplified value -- allowing us to move larger datasets into more practical settings;
    \item Experimentally show that this distance is stable  and more discriminative than the bottleneck distance while still retaining a similar ``largest feature difference" approach to similarity measuring;
    \item Show the usefulness of stable, discriminative distances on several datasets.
\end{itemize}

\section{Related Work}

\subsection{Graph-based Topological Descriptors}

Graph-based topological descriptors include merge trees (sometimes specifically referred to as split trees or join trees), Reeb graphs, contour trees, and mapper graphs. Each descriptor is designed to show the changes of the topological structure in the underlying dataset. Reeb graphs are, arguably, the most complex descriptor in this family.  Reeb graphs contract each component of each level set into a single point. The contour tree is simply the Reeb graph defined on a simply connected domain -- making the contour tree a well-defined tree rather than a directed multigraph. Merge trees are then the simplest (both in structure and in computational cost) of these in that it encodes the sublevel (or superlevel) set topology rather than the levelset topology.   

The visualization community has a long history of providing effective computations of these descriptors as well as using them for data analysis.  Heine et al.~recently surveyed many of their uses~\cite{Heine2016survey}.  In this section, we highlight some of the more recent works as they relate to applications of level set topology, rather than providing an exhaustive survey. Oesterling et al.~construct topological landscapes of high-dimensional point clouds using join trees~\cite{oesterling2011visualization}.  Bremer et al.~capture the behavior of turbulent mixing by developing hierarchical techniques for merge trees~\cite{bremer2010interactive}.  Thomas et al. explore symmetry detecting using contour trees~\cite{thomas2011symmetry}.  Widanagamaachchi et al.~study atmospheric phenomena by constructing a tracking on merge trees~\cite{widanagamaachchi2017exploring}.  Yan et al.~compute a structural average of merge trees for understanding statistical properties of collections~\cite{Yan2019}.

\subsection{Distances on Merge Trees}

Stability of merge trees was proven when the interleaving distance between merge trees was introduced \cite{Morozov2013}. Afterwards, functional distortion distance was introduced for Reeb graphs \cite{Bauer2014}, the interleaving distance was extended to Reeb graphs \cite{deSilva2016},  and several edit distances were introduced for Reeb graphs \cite{DiFabio2012,DiFabio2016,Bauer2020}. While the Reeb graphs are inherently different summaries of the scalar field, merge trees are still a 1-dimensional graph and thus many of the definitions introduced in these works can be applied directly to merge trees. Researchers have actually shown the equivalence of interleaving, functional distortion, and the universal distance on merge trees \cite{Bauer2021}. Each of these aforementioned distances have been proven to be both stable and discriminative to the bottleneck distance \cite{deSilva2016,DiFabio2012,DiFabio2016,Bollen2021,Bauer2020,Morozov2013,Bauer2014}.

Unfortunately, implementations of these distances have been scarce due to their computational complexity. In order to have distances which are practical, other researchers have focused their efforts on defining distances specifically on merge trees due to their simplicity. As stated before, these distances loosen the restriction on either stability or discriminativity in order to have distances which are computationally feasible. To avoid confusion, we will call the collection of distances consisting of the interleaving, functional distortion, and universal distance as the \textit{theoretical merge tree distances}. We call the collection of distinct distances we discuss below the \textit{experimental merge tree distances}.

Sridharamurthy et al. introduced an edit distance between merge trees which is experimentally shown to be more discriminative than both the bottleneck and 1-Wasserstein distance \cite{Sridharamurthy2018}. This distance loosens the restriction on stability which makes it computationally feasible. Cases of instability are still addressed by introducing an adjustable parameter which combines saddles which are within the parameters value -- simplifying the topology of the merge tree. The distance was also proven to be a well-defined metric on the space of merge trees. This idea was later expanded upon with the introduction of the local merge tree edit distance -- a well-defined metric specifically designed to study the local similarities at multiple resolutions rather than providing a global measure \cite{Sridharamurthy2021}.

Beketayev et al. provides a computation of a similarity measure for merge trees by first computing all of its branch decomposition trees -- data structures which encode features of the merge tree as nodes in a new tree -- and then finding pairwise matchings between these trees. The matching imposes a restriction that if $x$ matches to $y$ and $x'$ is a child of $x$, then $x'$ must be matched to a child of $y$ (or be deleted). This is similar to the `ancestor preserving' restriction imposed by standard tree edit distance (TED) \cite{KuoChung1979}. Our distance that we propose in \autoref{sec:distance} similarly uses branch decomposition trees in order to encode the feature of a merge tree. We divert from this work by removing the `ancestor preserving' restriction which allows us to maintain stability of our distance. 

Saikia et al. \cite{Saikia2014,Saikia2017} produces a similar distance to the one defined by Beketayev et al. They introduce a dynamic programming algorithm to create an \textit{extended branch decomposition tree} -- a data structure which encodes similar data to the conglomerate of all possible branch decomposition trees without having to store all these possibilities in memory. They show the application of this distance on self-similarity of scalar fields and detecting periodicity in time-varying datasets.

\section{Technical Background}
\label{sec:background}

\subsection{Scalar Fields and Merge Trees}

To ensure that our resulting structures are well-behaved, we elect to focus our attention towards piecewise linear-scalar fields: scalar fields in which the domain $\X$ is triangulable and the function $f$ is piecewise linear. Furthermore, we will focus our work on scalar fields in which the domain is a simply connected, two-dimensional manifold and the function $f$ is a simple Morse function \cite{CompTop2010}. These conditions ensure that the merge tree is a well-defined, one-dimensional graph \cite{deSilva2016}.

\begin{definition}
A \textbf{scalar field} (equivalently an \textbf{$\R$-space}) is a pair $(\X,f)$ where $\X$ is topological space and $f:\X \to \R$ is a continuous real-valued function.
\end{definition}

\begin{definition}
A \textbf{sublevel set} of $\X$ at $a\in \R$, denoted as $\X_a$ is the pre-image of the set $(-\infty,a]$ under $f$. Similarly, a \textbf{superlevel set} of $\X$ at $a$ is $f^{-1}[a,\infty)$ and is denoted as $\X^a$.
\end{definition}

\begin{definition}\label{def:mergeTree}
We define an equivalence relation $\sim_f$ on $\X$ by stating that $x \sim_f y$ if $x,y \in \X^a$ and $x$ and $y$ both lie in the same connected component of the superlevel set. We define $\X_f$ to be the quotient space $\X / \sim_f$ and define $\tilde{f}:\X_f \to \R$ to be the restriction of $f$ to the domain $\X_f$. The pair $\mathcal{S}_f := (\X_f,\tilde{f})$ is called the \textbf{split tree} of $(\X,f)$. The \textbf{join} tree $\mathcal{J}_f$ is defined analogously using sublevel sets rather than superlevel sets. The split and join tree make up the class of \textbf{merge trees}. We denote a general merge tree of the scalar field $(\X,f)$ as $\MM_f$.
\end{definition}

In what follows, we will be \textbf{working solely with the split tree}. Some definitions, theorems, and parts of our algorithm work for both the join and split tree, while others are specific to the split tree due to aspects such as increasing paths from saddle to extrema rather than decreasing paths. However, if we were to negate the original function defined on the scalar field, we can provide a distance for the join tree as well. To this end, we will elect to use the term \textbf{merge tree} and use the notation $\MM_f$ for a merge tree defined on a scalar field $(\X,f)$.

Since merge trees can be considered as labeled graphs, we will often denote the vertices and edges of $\MM_f$ as $V(\MM_f)$ and $E(\MM_f)$, respectively. From this, we can define the notion of \textbf{merge tree isomorphism}.

\begin{definition}
Two merge trees $\MM_f$ and $\MM_g$ are isomorphic if there exists a bijection $\alpha: V(\MM_f) \to V(\MM_g)$ such that 1) the edge $e(u,u') \in E(\MM_f)$ if and only if $e(\alpha(u),\alpha(u')) \in E(\MM_g)$ and 2) for every $u \in \MM_f$, we have $f(u) = g(\alpha(u))$. 
\end{definition}

\subsection{Persistence Diagrams and Bottleneck Distance}

Instead of defining the persistence diagram on the scalar field, we elect to define the persistence diagram by using using the merge tree as a scalar field itself since it reduces the number of classes of points in the persistence diagram and overall makes the comparison between distances on persistence diagrams and distances on merge trees simpler; see Bollen et al.\cite{Bollen2021} for a discussion on defining the persistence diagram of graph-based descriptors. For the sake of brevity, we show how to construct a persistence diagram from a merge tree and its properties rather than theoretical definitions.

The persistence diagram $\dgmf$ of a merge tree $\MM_f$ is a multiset of points $(a,b)$ which each represent a pair of vertices of the merge tree. These pairs intuitively represent different features of the merge tree. To determine which vertices are paired together, we introduce the \textbf{elder rule}.

\begin{definition}
The \textbf{elder rule} is a pairing scheme between saddles and extrema of a merge tree that says $x$ is paired with $y$ if there exists a monotone increasing path from $y$ to $x$ and if for all saddles $y'$ on the same path with $f(y') > f(y)$, they are paired with an extrema $x'$ such that $f(x') < f(x)$. The \textbf{persistence diagram} is a multiset $\dgmf$ where $(f(x),f(y)) \in \dgmf$ if $x,y$ are a saddle extrema pair based on the elder rule with the addition of the pair $(f(g_1),f(g_2))$ where $g_1$ is the global minimum and $g_2$ is the global maxima.
\end{definition}

\begin{figure}
    \centering
    \includegraphics[width=0.45\textwidth]{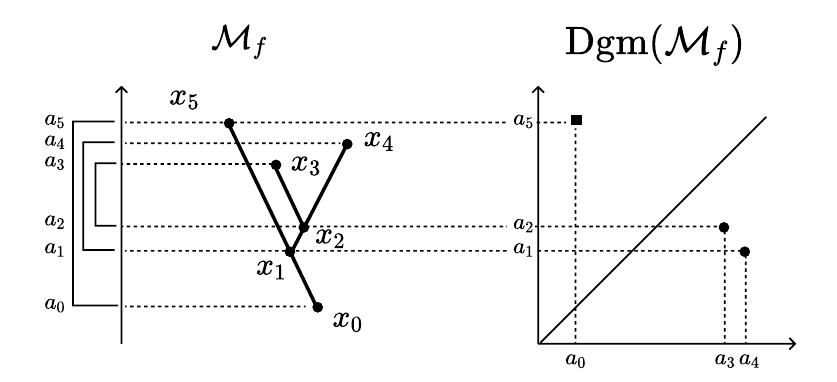}
    \caption{A merge tree $\MM_f$ with its accompanying persistence diagram $\Dgm(\MM_f)$. The square point of $\Dgm(\MM_f)$ represents the pairing of the global min and global max.}
    \label{fig:persDgm}
\end{figure}

Persistence diagrams have been shown to be stable under the well-studied bottleneck distance \cite{Steiner2005,Steiner2009}.
The bottleneck distance assign a \textbf{cost} to a matching between the points of two persistence diagrams. We allow for each point to also be matched to an \textbf{empty node} which can be thought of as deleting or inserting that feature.

\begin{definition}[\textbf{Bottleneck Distance}]
\label{def:bottleneckDistance}
Let $D_1,D_2$ be two persistence diagrams and let $\lambda$ denote an \textbf{empty node}. We define $\bar{D_i}:=D_i\cup\{\lambda\}$. A \textbf{matching} $M$ between $D_1$ and $D_2$ is a binary relation $M \subseteq \bar{D}_1 \times \bar{D}_2$ such that each element from $D_1$ and $D_2$ appear in exactly one pair $(x,y) \in M$.

The \textbf{cost} of a pair $(x,y)\in M$ is defined as
\[c(x,y) = \begin{cases} 
      \max\{|x_1-y_1|,|x_2-y_2|\} & x\in D_1,y\in D_2 \\
      \frac12|x_1-x_2| & x\in D_1, y = \lambda \\
      \frac12|y_1-y_2| & x = \lambda, y \in D_2 
   \end{cases}
\]
The \textbf{cost of a matching $M$}, denoted as $c(M)$, is then the largest cost of all pairs in the matching.
\end{definition}

\subsection{Branch Decomposition Trees}
The branch decomposition tree (BDT) is a data structure which, in topological data analysis, attempts to pair the saddles of contour trees or merge trees to the extrema of that tree. Each node in the BDT would then represent a \textit{feature} of the original scalar field. Pascucci used the branch decomposition trees to inform a layout for complex contour trees with many self-intersections~\cite{pascucci2009toporrery}. Since then, BDTs have seen additional use as representations of merge trees for their comparison \cite{Saikia2017,Beketayev2014}.

Each merge tree or contour tree has precisely $2^{\frac{n}{2}-1}$ different possible BDTs, where $n$ is the number of nodes \cite{Beketayev2014}. Often, a unique BDT is constructed by weighting the choice of pairing based on a particular measurement -- such as persistence of the branch or the number of voxels of the branch in the scalar field \cite{Saikia2017}.

\begin{definition}
A \textbf{branch} is a monotone (in function value) path traversing a sequence of nodes in the merge tree $\MM_f$. The first and last nodes of this sequence are called the \textbf{endpoints} of the branch.
\end{definition}

\begin{definition}
A \textbf{branch decomposition} of a merge tree is a set of branches such that every edge $e \in E(\MM_f)$ appears in exactly one branch.
\end{definition}

\begin{definition}
A branch decomposition of a merge tree is a \textbf{hierarchical decomposition} if (1) there is exactly one branch which connects two extrema to one another (called the \textbf{root branch} and (2) every other branch connects an extrema to a node that is interior to another branch.
\end{definition}

\begin{definition}
Let $H_f$ be a hierarchical decomposition of a merge tree $\MM_f$. The \textbf{branch decomposition tree (BDT)}, $b_f$, with respect to $H_f$ is a rooted tree $b_f = (V,E)$ where $v = (v_1,v_2) \in V$ represent the branches of $H_f$. The edge $e(v,u) \in E$ if and only if $u$ has an endpoint interior to the branch $v$.
\end{definition}

For every hierarchical decomposition of a merge tree $\MM_f$, we obtain a unique BDT $b_f$. Each node $u\in b_f$ corresponds to two vertices of $\MM_f$. If $u \in b_f$ is not the root node, then there is a corresponding saddle $u_s \in \MM_f$ and a corresponding maxima $u_e \in MM_f$. The root node $r \in b_f$ corresponds to the global minimum $r_s \in \MM_f$ and a maxima $r_e \in \MM_f$. We denote the set of all possible branch decomposition trees of a merge tree $\MM_f$ as $\BB_f$. \autoref{fig:mergeTreeToBDT} shows the eight different BDTs for a merge tree with eight nodes.
\begin{figure}
    \centering
    \includegraphics[width=0.49\textwidth]{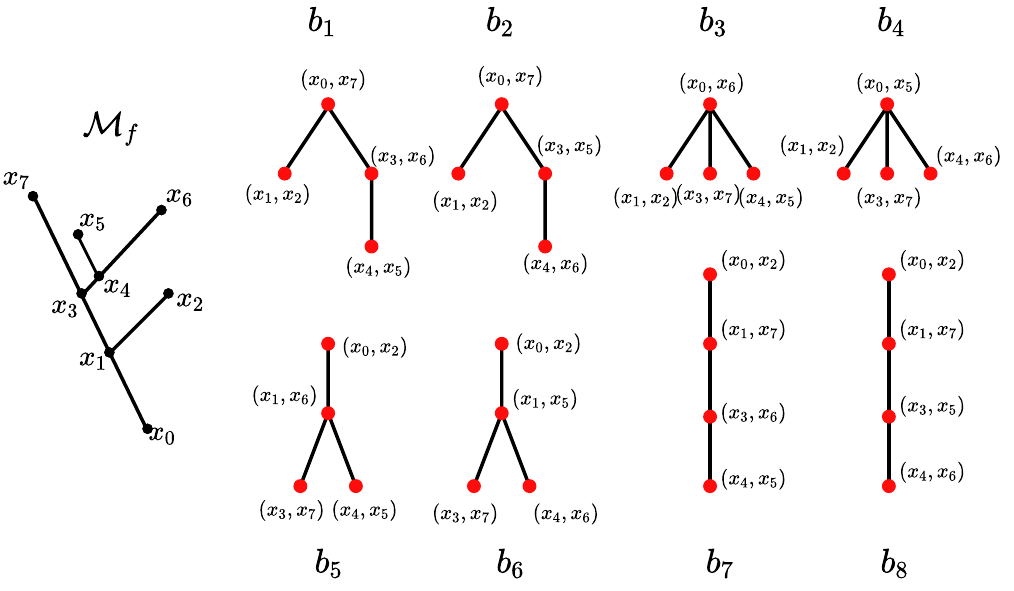}
    \caption{Set of possible branch decomposition trees for a single merge tree $\MM_f$.}
    \label{fig:mergeTreeToBDT}
\end{figure}
\subsection{Stability and Discriminativity}

\begin{definition}
A distance $d$ defined merge trees is said to be \textbf{stable} if and only if 
\[d(\MM_1,\MM_2) \leq ||f-g||_{\infty},\]
where $f,g$ are the corresponding functions for the scalar fields of $\MM_1$ and $\MM_2$, and $\MM_1$ and $\MM_2$ are defined on the same domain $\X$.
\end{definition}

Stability of a distance guarantees that point-wise perturbations introduced into the dataset will not drastically change the merge tree. Saikia et al. \cite{Saikia2014} defined two different types of instabilities which are exhibited in merge trees: \textbf{horizontal instabilities} and \textbf{vertical instabilities}.
\begin{definition}
Let $(\X,f)$ be a scalar field with respective merge tree $\MM_f$ such that 
there exists a pair $s_1,s_2\in V(\MM_f)$, with $\deg(s_1) = \deg(s_2) = 3$ and such that $|\tilde{f}(s_1) - \tilde{f}(s_2)| < 2\e$. If $e(s_1,s_2) \in E(\MM_f)$, then $(\X,f)$ is \textbf{horizontally} $\mathbf{\e}$\textbf{-unstable}.
\end{definition}

\begin{definition}
Let $(\X,f)$ be a scalar field with respective merge tree $\MM_f$ such that there exists a pair of vertices $m_1,m_2\in V(\MM_f)$, with $deg(m_1) = deg(m_2) = 1$ and such that $|\tilde{f}(m_1) - \tilde{f}(m_2)| < 2\e$. Let $(s_1,m_1) \in \Dgm(\MM_f)$ and $(s_2,m_2) \in \Dgm(\MM_f)$ be the persistence pairs corresponding to $m_1$ and $m_2$, for some $s_1,s_2 \in V(\MM_f)$. If there exists monotone paths $p_{1,2}:s_1 \to m_2$ and $p_{2,1}: s_2 \to m_1$, then $(\X,f)$ is \textbf{vertically} $\mathbf{\e}$\textbf{-unstable},
\end{definition}

For our algorithm, we will use branch decomposition trees to organize the features of the scalar field. Due to this setup, we may have a situation where a small change in the function values of extrema switches their total ordering -- possibly altering the topology of BDTs. This is called a \textbf{vertical instability}. \autoref{fig:instabilities} depicts the affects of perturbations in a vertical and horizontal manner on a function $f$.

\begin{figure}
    \centering
    \includegraphics[width=0.5\textwidth]{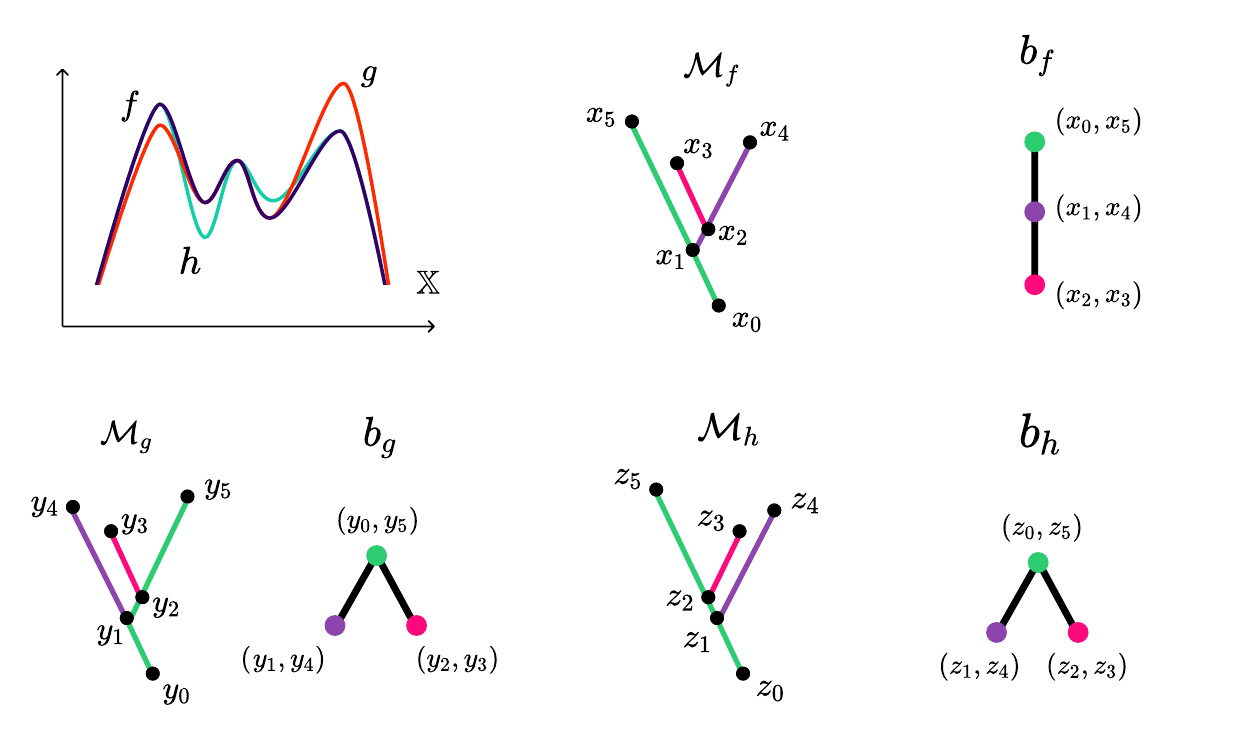}
    \caption{Three functions, $f,g,h$ all defined on the same domain $\X$. The functions $g$ and $h$ are perturbed version of $f$, where $g$ presents a vertical instability and $h$ presents a horizontal instability. The corresponding branch decomposition trees are the unique BDTs determined by the persistence of each feature.
    \label{fig:instabilities}}
\end{figure}

\begin{definition}
A distance $d$ defined on merge trees is said to be more \textbf{discriminative} than a baseline distance $d_0$ if there exists some constant $c > 0$ such that \[d_0(\MM_f,\MM_g) \leq c\cdot d(\MM_f,\MM_g),\] for all merge trees $\MM_f,\MM_g$, and if there does not exist a constant $c'$ such that $d_0 = c'\cdot d$.
\end{definition}

If a similarity measure is strictly bounded below by a baseline distance (up to a constant $c$), then there are cases in which the baseline distance is not able to discern between two merge trees while the distance $d$ does detect some dissimilarity. Furthermore, this implies that if the distance $d$ detects no difference between two merge trees, then the baseline will not detect any difference as well. 

A core position on discriminativity being desirable is that we expect these merge tree distances to inherently be more computationally complex than persistence diagram distances since the graph-based descriptors are strictly more complex than set-based descriptors. Thus, these merge tree distances will trade off their computational efficiency for encoding more information in the similarity measure. As with the theoretical merge tree distances, we will use the bottleneck distance as our baseline since the theoretical merge tree distance and bottleneck distance all use a ``max-feature-difference'' approach to similarity measuring.

A similar notion to discriminativity is \textbf{isomorphism invariance}.
\begin{definition}
A distance $d$ on merge trees is \textbf{isomorphism invariant} if $d(\MM_f,\MM_g) = 0$ if and only if $\MM_f$ and $\MM_g$ are merge tree isomorphic.
\end{definition}

It has been shown that the bottleneck distance is not isomorphism invariant on the space of merge trees while all of the theoretical graph-based distances are~\cite{Bollen2021}. The merge tree edit distance is also isomorphism invariant~\cite{Sridharamurthy2018}.

\subsection{Zigzag Diagrams}

The \textbf{universal distance} (originally referred to as the \textbf{Reeb graph edit distance} \cite{Bauer2020}) is a stable, discriminative distance defined on Reeb graphs and merge trees which has been shown to be the largest stable distance defined on Reeb graphs -- a property known as \textbf{universality}. On merge trees, it was shown to be equivalent to the interleaving and functional distortion distance \cite{Bauer2021} -- thus making all of these distances universal.

The universal distance is defined by constructing a \textbf{zigzag} diagram of topological spaces which connects a source merge tree $\MM_f$ to its target $\MM_g$. 
These zigzag diagrams can be intuitively thought of as a sequence of operations carrying one merge tree to another. We use a simplified version of the zigzag diagram for use with our distance. The term \textit{carry} is used to state that we are transforming a source merge tree $\MM_f$ into a merge tree $\MM'$ which is isomorphic to $\MM_g$.

\begin{definition}
Let $\MM_f,\MM_g$ be two merge trees. A \textbf{zigzag diagram} $Z$ is a sequence of merge trees $M = \{\MM_f = \MM_1,\MM_2,\ldots,\MM_{n-1},\MM_n = \MM_g\}$ coupled with a sequence of 1-dimensional graphs $X =\{X_1,\ldots,X_{n-1}\}$ such that for each $X_i$, there are two valid maps $q_{i,i}:V(X_{i}) \to V(\MM_{i}) ,q_{i,i+1}: V(X_{i})\to V(M_{i+1})$ which respect edge assignments. That is, if $e(x_j,x_k) \in X_i$, then $e(q_{i,i}(x_j),q_{i,i}(x_k)) \in \MM_i$. The sequence $X$ will be called the \textbf{connecting spaces} of $Z$ while $M$ is called the \textbf{merge trees} of $Z$
\end{definition}

Each merge tree $\MM_i\in M$ will have an associated function $f_i$. In general, these merge trees need not be Morse. Specifically, we will have merge trees with vertices of degree 4 which is not permitted under the definition of Morse functions on 2-manifolds. The connecting spaces are responsible for changing adjacencies in the merge trees. \autoref{fig:zigzagBackForth} depicts two different zigzag diagrams. The first carries a merge tree $\MM_f$ to $\MM_g$, while the bottom diagram is in reverse order.

\begin{definition}
The \textbf{limit} $\dLimit$ of a zigzag diagram $Z$ with $n$ merge trees is the $(n-1)$-dimensional space where $x = (x_1,\ldots,x_{n-1}) \in \dLimit$ if $q_{i,i+1}(x_i) = q_{i+1,i+1}(x_{i+1})$ for all $x_i$. We define $y_i$ as $y_i = q_{i,i}(x_i)$.
\end{definition}

\begin{definition}
The \textbf{spread} $\spread$ of an element $x \in \dLimit$ is the difference between the maximum function value and minimum function value it attains in the zigzag diagram. That is,
\[\spread(x) = \max_{i=1,\ldots,n}f_i(y_i)-\min_{i=1,\ldots,n}f_i(y_i).\]
The \textbf{cost} of the zigzag diagram $Z$ is then the largest spread of its limit $\dLimit$.
\[c(z) = \max_{x\in \dLimit} \spread(x).\]
\end{definition}

When the choice of zigzag diagram is not clear, we use the notation $\dLimit(Z)$ to denote the limit of the zigzag diagram $Z$. See Appendix A.2 for another example of a zigzag diagram with the corresponding spread.

\section{Merge Tree Matching Distance}
\label{sec:distance}

The bottleneck distance constructs a similarity measure between scalar fields by 1) effectively encoding the features of the scalar field as a multiset of points, 2) constructing a way to match the encoded features of one scalar field to the features of another, 3) computing a cost on this matching by computing the largest difference between two matched features, and 4) taking the distance to be the lowest cost over all possible matchings. The theoretical merge tree distances can be thought of in a similar fashion. For example, the universal distance requires a choice of which features to transform into others, provides a way to carry out this transformation by using zigzag diagrams, and then computes a cost of this zigzag diagram \cite{Bauer2020}.

Our distance is motivated by three main objectives: 
\begin{itemize}
\item create a distance which is similar to the bottleneck distance and the theoretical merge tree distances in that it 1) properly encodes the features of the merge tree, 2) matches features of one merge tree to another, 3) assigns a cost to this matching by computing the largest feature difference, and 4) minimizes this cost over all possible matchings;
\item construct it in such a way that it is isomorphism invariant on the set of merge trees as well as being more discriminative than the bottleneck distance; and
\item make sure that the distance handles cases of instability correctly.
\end{itemize}

\subsection*{Distance Definition}

When constructing a distance between merge trees, we need to make sure that the distance captures the difference based on the relationship between features that are in the original scalar field rather than solely on the paired critical points. We encode this hierarchical relationship between features using the BDT, which also captures topological features (pairs of critical points) as individual vertices in the BDT.

Unlike persistence diagrams, there are many different BDTs for each merge tree. Using only one can lead to vertical instabilities in the distance. Thus, in order to adequately find the distance between two merge trees, enumeration of all the BDTs is needed, similar to Beketayev et al.~\cite{Beketayev2014}. Let $\MM_f$, $\MM_g$ be two merge trees with respective sets of BDTs $\BB_f,\BB_g$. A matching between a fixed $b_f \in \BB_f$ and $b_g\in\BB_g$ gives us a matching between the features of $\MM_f$ and $\MM_g$.

\begin{definition}
Let $\MM_f,\MM_g$ be two merge trees with respective sets of BDTs $\BB_f,\BB_g$. Let $b_1 \in \BB_f,b_2\in\BB_g$ be two BDTs and let $\lambda$ be an \textbf{empty node} not in $V(b_1)$ nor $V(b_2)$. We let $x_e,x_s$ denote the extrema and saddle nodes of a vertex $x \in V(b_i)$. We define $\bar{b}_i:=b_i\cup\{\lambda\}$. A \textbf{matching} $M$ between $b_1$ and $b_2$ is a binary relation $M \subseteq \bar{b}_1 \times \bar{b}_2$ such that the following conditions hold:
\begin{enumerate}
    \item $(r_1,r_2) \in M$, where $r_1 \in V(b_1),r_2 \in V(b_2)$ are the respective roots of $b_1,b_2$.
    \item Each element in $V(b_1)$ and $V(b_2)$ appear in exactly one pair in $M$.
    \item If $(x,\lambda) \in M$, then $(f(x_e),f(x_s)) \in \dgmf$.
    \item If $(\lambda,y) \in M$, then $(g(y_e),g(y_s)) \in \dgmg$ .
\end{enumerate}
A \textbf{partial matching} $M'$ is a matching between BDTs with condition 1) loosened to have each element of $V(b_1)$ and $V(b_2)$ appear in at most one pair of $M'$.
\end{definition}

Elements of a matching $M$ fall into three different categories: \textbf{insertion} pairs which have the form $(\lambda,v)$, \textbf{deletion} which pairs have the form $(u,\lambda)$, and \textbf{relabel} pairs which have the form $(u,v)$. If $(u,v) \in M$ and $(u_p,v_p) \notin M$, where $u_p,v_p$ are the parents of nodes $u,v$, then $(u,v)$ is further categorized as a \textbf{movement} relabel pair.

An \textbf{induced zigzag diagram} is the zigzag diagram which arises from carrying one merge tree to another. This induced zigzag diagram follows the protocol that we apply insertions, non-movement relabels, movement relabels, and then deletions. We apply insertions first and deletions last since we cannot create disconnected merge trees in our zigzag diagram. 

Each matching $M$ between two BDTs $b_i \in \BB_f$,$b_j\in\BB_g$ induces two zigzag diagrams: the forward zigzag diagram $Z_{f,g}$ from $\MM_f$ to $\MM_g$, and the backward zigzag diagram $Z_{g,f}$ from $\MM_g$ to $\MM_f$. The difference in these two induced zigzag diagrams is when the relabel pairs are applied. Taking the minimum cost over the forward and backward zigzag diagram ensures that we handle the instability in a way that keeps the distance below the $L_{\infty}$ distance. \autoref{fig:zigzagBackForth} depicts an example of the possible difference. Since these merge trees are horizontal $\e$-unstable, we need to make sure that our distance is less than or equal to $\e$ for us to have a distance below the $L^{\infty}$ distance. In this case, $\e = |\tilde{g}(y_1)-\tilde{g}(y_2)|$. Note that the backward zigzag diagram has a spread less than $\e$.

We define the cost of a matching as follows:

\begin{definition}
The \textbf{cost} of a matching $M$ is the minimum value between the costs of the induced zigzag diagrams. That is,
\[c(M) = \min_{Z \in \{Z_{f,g},Z_{g,f}\}}c(Z) = \min_{\dLimit \in \{\dLimit(Z_{f,g}),\dLimit(Z_{g,f})\}}\max_{x\in\dLimit(Z)}\spread(x)\]

\end{definition}

\begin{definition}
The \textbf{merge tree matching distance} is defined as the minimum cost of all matchings between all pairs of BDTs. That is,
\[d_M(\MM_f,\MM_g) = \min_{b_i \in \BB_f,b_j \in \BB_g}\min_{M\in \mathsf{M}_{i,j}}c(M),\]
where $\mathsf{M}_{i,j}$ denotes the set of all possible matchings between $b_i\in\BB_f$ and $b_j\in\BB_g$.
\end{definition}

\begin{figure}[h]
    \centering
    \includegraphics[width=0.5\textwidth]{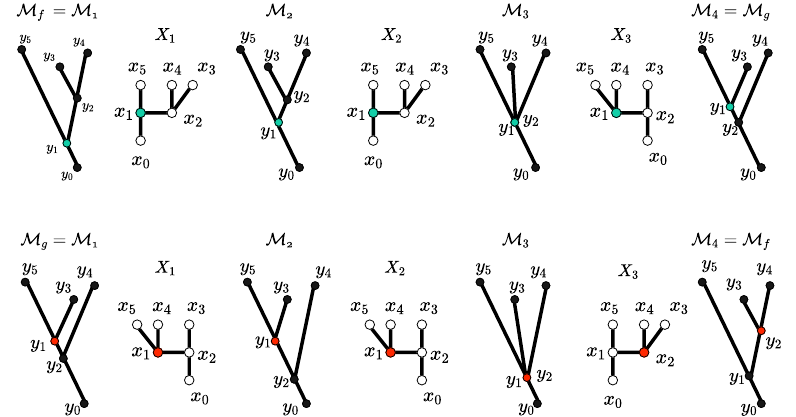}
    \caption{Depiction of the forward zigzag diagram $Z_{f,g}$ (top) and the backward zigzag diagram $Z_{g,f}$ (bottom) of a matching between merge trees. The connecting spaces are viewed as 1-dimensional graphs in between each of the merge trees of the sequence. Here, we always have the mapping $x_i \to y_i$ for all maps from the connecting spaces to the respective merge trees. With horizontal instabilities, it is desirable that moving the branch $e(y_2,y_3)$ has a cost equal to the largest difference between moving $y_2$ up to its final position or $y_1$ down to its final position. $Z_{f,g}$ achieves this cost, while $Z_{g,f}$ achieves a larger cost due to relabeling the saddle downwards to begin with, and then performing the movement.}
    \label{fig:zigzagBackForth}
\end{figure}

\begin{proposition}
The merge tree matching distance is isomorphism invariant. That is, $d_M(\MM_f,\MM_g) = 0$ if and only if $\MM_f$ and $\MM_g$ are merge tree isomorphic.
\end{proposition}
\begin{proof}
Suppose $\MM_f$ and $\MM_g$ are isomorphic. Then, each have identical sets of BDTs. Let $b_f,b_g$ be two such identical BDTs of $\MM_f,\MM_g$, respectively. We define $M$ to be the matching induced by the isomorphism between $b_f$ and $b_g$. Since they are isomorphic trees, there are no movement relabels in the induced zigzag diagram. Then, the induced zigzag diagrams will both consist of only the single merge tree since $\MM_f$ and $\MM_g$ are already isomorphic to one another.
Now, suppose that the merge tree matching distance between $\MM_f$ and $\MM_g$ is 0. Then, there must exist a pair of BDTs $b_f$ and $b_g$ such that the matching induces no deletions, insertions, movement relabels, or relabels which incur a cost. Thus, $b_g$ and $b_g$ must be tree isomorphic as well as having the same values on each of their corresponding saddles and extrema. Thus, $\MM_f$ and $\MM_g$ must also be isomorphic.
\end{proof}

\section{Algorithm}
\label{sec:algorithm}

Our algorithm is dependent on finding the cheapest matching between the vertices of all possible BDTs of two input merge trees. We are motivated by the well-studied \textbf{graph edit distance (GED)} in order to solve this. As stated in Section \ref{sec:distance}, the BDTs allow us to convert merge trees into data structures where the features of interest are now single vertices rather than pairs of vertices. GED has a similar problem statement: given two graphs $G_1$ and $G_2$, find a matching between their vertices by deducing an edit sequence between the graphs. From this edit sequence a cost is computed.

Our algorithm is split into three main components:
\begin{enumerate}
    \item Construction of all BDTs from the input merge trees $\MM_f$,$\MM_g$.
    \item Finding full matchings between two BDTs $b_f$,$b_g$ using the A* algorithm.
    \item Computing the cost of a full matching $M$ between $b_f$ and $b_g$ by first constructing the elements of the limit $\dLimit(Z)$.
\end{enumerate}

\subsection{Constructing All Branch Decomposition Trees}
\label{subsubsec:constructBDTs} Suppose we have a merge tree $\MM_f$. As an example, we describe how to construct a persistence-based BDT. We first take the global min $v_0$ and find the path $p_{v_0,v_n}$ to the global max $v_n$. The pair $(v_0,v_n)$ is placed as the root of the BDT. For every vertex $v$ along this path, we recursively call the construction of the BDT algorithm with $v$ being the new root. Each node $v \in p_{v_0,v_n}$ contributes a single child node to the root node. The algorithm ends when all nodes are paired.

For general BDTs, we must consider every maxima $v$ that is in the up-path of a root $u$ as a possible pairing rather than just the global maxima of that branch. In order to avoid redundant computations of the same branches, we add another recursive layer to the BDT algorithm. 

Suppose that $(v_s,v_e)$ is a pair corresponding to a branch $\mathsf{B}$ which has just been placed into a BDT $b$. For computation of a single BDT, we would recursively call this operation for all $v \in p_{v_s,v_e}$. However, we know that for every node $(v_s,v_e)$, there are multiple different configurations of its children. To avoid redundant computations, we generate a specific number of BDTs for each possible pairing of a saddle $v$, for each $v \in p_{v_s,v_e}$. First, assume $v$ is the only vertex in $p_{v_s,v_e}$ which is not one of the endpoints. Let $\{u_1,\ldots,u_m\}$ be the set of extrema in the up-path of $v$. For each pair $(v,u_i)$, we construct a new copy of $b$ and add $(v,u_i)$ as the child of $(v_s,v_e)$.

Now, suppose instead that $\{v_1,\ldots,v_k\}$ is the set of vertices in $p_{v_s,v_e}$ which are not the endpoints. We proceed to carry out the same step as if there was only one non-endpoint vertex $v$, except instead of making $m$ copies of $b$, we need to make a copy for the possible combinations of choices of extrema pairing for all $v_i \in \{v_1,\ldots,v_k\}$. More specifically, if $m_i$ denotes the number of extrema that $v_i$ can be paired with, we make $m_1\cdot m_2\cdot\ldots\cdot m_{k-1}\cdot m_k$ copies of $b$. Each $b$ gets a unique set of nodes added as the children of $(v_s,v_e)$.

\subsection{Finding the Best Matching Using A*}
\label{subsubsec:a*}

The A* algorithm maintains a priority queue $Q$ which holds a list of partial matchings $M'$ with a current cost $c(M')$. Since our distance requires the full matching before a true cost can be determined, we use the bottleneck distance between the currently matched nodes in order to under approximate cost. More specifically, let $P$ be a partial matching between $b_f$ and $b_g$. Then $c(M') = \max_{(u,v) \in M'} c(u,v)$, where $c(u,v)$ is the cost function in \autoref{def:bottleneckDistance}. As stated earlier, providing an under approximation to the true cost will guarantee that we still reach an optimal solution. However, since our under approximation will not necessarily converge to the true cost when we reach a full matching (unlike standard GED), we have to have an additional step which computes the true cost of the matching and to determine whether to continue finding better matchings dependent on this cost. Appendix A.3 shows an overview of our A* algorithm with this additional module.

The efficiency of A* is heavily dictated by introducing pruning techniques to reduce the number of possible matchings and using a good heuristic function which approximates the future cost along a particular path in the search tree. In our case, we can effectively prune the search space by not matching two nodes to one another if it would be cheaper to simply insert or delete both the nodes. As for a heuristic function, we introduce the function $h(M')$ which finds the lowest possible future cost based on the fact that if $b_f$ has $n$ unmatched nodes and $b_g$ has $m$ unmatched nodes, we must insert or delete the difference in the number of nodes. See Appendix A.1 for more information on our pruning and heuristic function. As opposed to GED, our ``approximate cost" would then be $\max\{c(M'),h(M')\}$ rather than $c(M') + h(M')$.

\subsection{Computing the Cost by Constructing the Limit $\dLimit(Z)$}\label{subsec:findingLimit}

Let $b_f$ and $b_g$ be the branch decomposition trees that we are comparing and let $M$ be the current full matching between them. As stated before, the forward zigzag diagram $Z_{f,g}$ is constructed by applying insertions, non-movement relabels, movement relabels, and then deletions to $b_f$ which ultimately creates $b_g$. Just as the graph edit distance can be thought of as constructing edit sequences to carry one graph to another, our distance can be thought of as altering the branch decomposition $b_f$ with this set of operations in order to construct $b_g$. 

There is a one-to-one mapping between the vertices of any two connecting spaces. Thus, we keep the labeling of each connecting space the same. When we say that $x_i = x_j$ for two vertices in different connecting spaces, this implies that $i = j$. We denote the set of vertices in the connecting spaces as $V(X)$.

Let $x^i,x^{i+1}$ be vertices of $X_i,X_{i+1}$, respectively. Note that $x_i = \{x_i^1,\ldots,x_i^{n-1}\}$ is an element of $\dLimit(Z)$ for all $x_i \in V(X)$. Additional elements are added if $q_{i,i+1}(x^i) = q_{i+1,i+1}(x^{i+1})$, for some $x^i \in X_i$ and $x^{i+1} \in X_{i+1}$ where $x^i\neq x^{i+1}$. We call these \textbf{swaps}. For example, Fig. \ref{fig:zigzagBackForth} has a swap in both the forward and backward zigzag diagram. For the forward zigzag diagram, a swap occurs in $\MM_3$ where $q_{1,2}(x_1) = q_{2,2}(x_2)$. This implies $\{x_1,x_1,x_2\}$ is also an element of the limit. For every swap, there are two options to continue constructing the limit: continue with the same vertex or move to the swapped vertex.

In what follows, each pair $(u,v) \in M$ would create at least one connecting space and one merge tree. If $(u,v)$ is a movement relabel, we may create more. Each time we would create a merge tree and connecting space, instead we create a copy the previous BDT and alter it according to the pair (inserting a node, deleting a node, relabeling, or moving a node). By the end of the algorithm, we are left with a set of BDTs $\{b_f = b_1,b_2,\ldots,b_{m-1},b_m = b_g\}$ and a list $S$ of length $m$ which will contain our swaps. These two data structures are enough to determine the elements of $\dLimit(Z)$.

To begin, we construct a copy of $b_f$ and apply all $n$ insertions. We sort the list of insertions by increasing depth order to make sure that if $u$ is to be inserted on $u'$, then $u'$ already exists. Non-movement relabels are then conducted in no particular order. 

Let $b_k$ be the current BDT. In a movement relabel, we have a source branch $u \in b_k$, its parent $u_p \in b_k$, a target branch $v \in b_g$ and the parent of the target $v_p \in b_g$. Let $u'$ be such that $(u',v) \in M$ and $u'_p$ be such that $(u'_p,v_p) \in M$. Our goal is then to move $u$ to have parent $u'_p$. We find the path $\rho: u_p \rightsquigarrow u'_p$. The node closest to the root of $b_k$ in depth is known as the \textbf{intersection}. The intersection is the branch where $u$ does \textit{not} need to attach to its saddle. For every other branch $w$ in the path $\rho$, we know we must add an additional connecting space and add the swap tuple $(u_s,w_s)$, where $u_s$ and $w_s$ are saddles of the respective merge trees, to the list. \autoref{fig:movementExample} depicts an example of recording the intersection and swaps for moving a branch.

\begin{figure}
    \centering
    \includegraphics[width=0.48\textwidth]{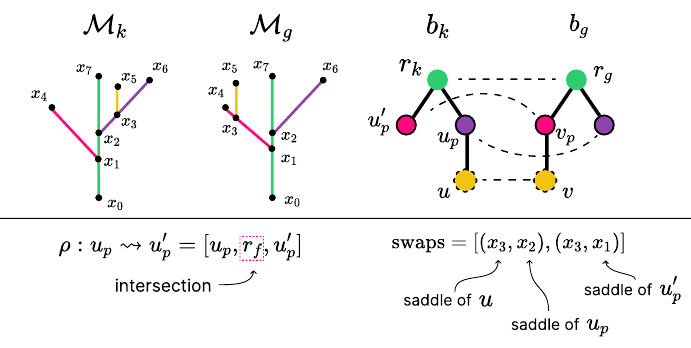}
    \vspace{-1em}
    \caption{A depiction of a movement and how we record the swaps and intersection. The branch labeled $u\in b_k$ is matched with the branch labeled $v\in b_g$, but their parents are not matched with each other. We find the path from the parent of $u$, denoted as $u_p$, to the node which matches with the parent of $v$, denoted as $u'_p$. This requires us to have the branch $u = (x_3,x_5)$ to pass the saddle $x_2$ and then pass $x_1$. Note that $r_f$ is the intersection and thus its saddle is not recorded as a swap.}
    \label{fig:movementExample}
\end{figure}

Movement relabels need to be conducted in an order that makes sure that the BDT stays connected. Our A* algorithm has the restriction that if $u$ is an ancestor of $u'$, then $u$ cannot become a descendant of $u'$ since it is suboptimal (i.e. another choice of BDT should be used if this is the case). However, once we begin altering the original BDT, there are situations where this $u$ needs to be moved onto its parent. To alleviate this, we maintain another priority queue which holds all the movements, first ranked by the depth to the root. If we were to apply a movement, we check if this disconnection will occur. If so, we push the movement pair back into this priority queue with a lower priority index and move onto the next movement. 

\subsection{Putting It All Together}
Algorithm~\ref{alg:mtmd} shows the pseudocode of how we first choose pairs of branch decomposition trees and then subsequently feed these pairs into our A* algorithm. Note that we introduce a ``cutoff'' variable which indicates when to stop the A* computation of a pair of BDTs. If the cheapest current matching is ever larger than this cutoff, we stop the computation and move onto the next pair.
\begin{algorithm}
\caption{mergeTreeMatchingDistance}\label{alg:mtmd}
\hspace*{\algorithmicindent} \textbf{Input:} Two merge trees $\MM_f,\MM_g$ \\
 \hspace*{\algorithmicindent} \textbf{Output:} Merge tree matching distance between $\MM_f,\MM_g$
\begin{algorithmic}[1]
\State $\BB_f$ = ConstructBDTs($\MM_f)$, $\BB_g$ = ConstructBDTs($\MM_g)$, cutoff = 0
\For{$b_f \in \BB_f$}
    \For{$b_g \in \BB_g$}
        \State currCost,completeMatch = aStar($b_f$,$b_g$,cutoff)
        \If{completeMatch == \texttt{true}}
            \State cutoff = currCost
            \State dist = currCost
        \EndIf
    \EndFor
\EndFor
\State \Return dist

\end{algorithmic}
\end{algorithm}

\section{Persistence Simplification}

Suppose we have two merge trees $\MM_f$ and $\MM_g$ whose distance $d_M$ is $A$. We can persistence simplify each by some $0 <\e < A$ to reduce its size -- providing a graph which will be more readily computable by both direct computation. Our distance has the convenient property that a simplification by $\e > 0$ will increase the distance by at most $\frac{1}{2}\e$. 

\begin{theorem}
Let $\MM_f,\MM_g$ be two merge trees whose distance $d_M$ is $A$ and let $0<\e<A$ be fixed. Then
\[A\leq d_M(P_{\e}(\MM_f),P_{\e}(\MM_g))\leq A+\tfrac{1}{2}\e.\]
\end{theorem}
\begin{proof}
Let $b_f,b_g$ be the optimal branch decomposition choices along with the optimal matching $M$. Suppose $u=(u_s,u_e) \in V(b_f)$ and $|f(u_s)-f(u_e)| < \e$. Then, persistence simplifying $\MM_f$ by $\e$ removes $u$ from $b_f$ and removes the pair $(u,*) \in M$, where $*$ may be the empty node $\lambda$ or some node $v \in V(b_g)$. If $(u,\lambda) \in M$, then removal of the pair via simplification will not increase the distance since $\frac12|f(u_s)-f(u_e)| < \e < A$ and therefore cannot be the largest cost pair in $M$. Now, suppose $(u,v) \in M$. If $|g(v_s)-g(v_e)| > \e$, then the cost may change depending on where $v$ is now assigned. In the worst case scenario, we delete $v$. In this case, the largest increase from deletion of $v$ and assigning $u$ to $v$ comes when they share a midpoint. Thus, the difference in deletion of $v$ to the relabel is at most $\frac12\e$. 
\end{proof}

Simplification by the same value of $\e$ is not necessary to achieve a bound on the distance. There may arise situations in which simplification by the same value of $\e$ yields merge trees with too little information for data analysis. The less features that exist in the merge tree means the less information we may glean from the matching provided by the distance.

\begin{corollary}
Let $\MM_f,\MM_g$ be two merge trees whose distance $d_M$ is $A$ and let $0<\e_1<A$ and $0 < \e_2 < A$ be fixed. Then
\[A\leq d_M(P_{\e_1}(\MM_f),P_{\e_2}(\MM_g))\leq A+\tfrac{1}{2}\max\{\e_1,\e_2\}.\]
\end{corollary}

Note that since the merge tree matching distance is always bounded below by the bottleneck distance, we can compute the bottleneck distance between the two merge trees and use the resulting value to gauge the value of $\e$.

\begin{figure*}
    \centering
    \includegraphics[width=0.97\textwidth]{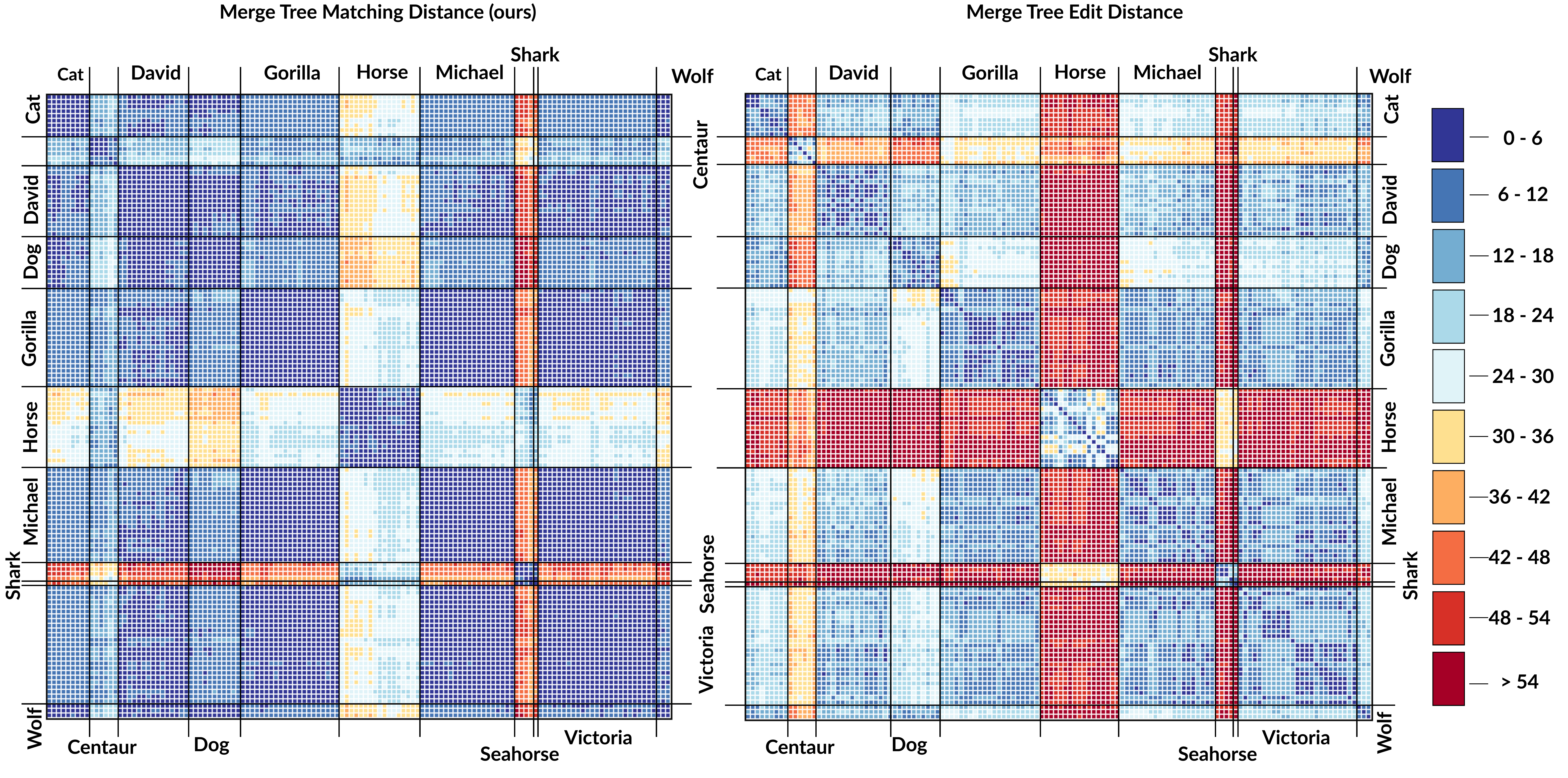}
    \vspace{-1em}
    \caption{Two distance matrices of computing distance on the average geodesic scalar field of the TOSCA non-rigid world dataset. The left is using our merge tree matching distance. The right is the merge tree edit distance from Sridharamurthy et al. \cite{Sridharamurthy2018}.}
    \label{fig:shapesDistMatrix}
\end{figure*}

\section{Experiments}
\label{sec:experiments}

We implemented the algorithm described above using Python. We use the Topology Toolkit \cite{ttk} to visualize and extract merge trees from the input datasets. The bottleneck distance was computed in python through Persim \cite{Persim}. Each experiment was run using on a single AMD EPYC 7642 machine at 2.4GHz using 32 of the cores.  We split the distance computations in batches to work on subsets of the data in parallel.  Within a batch, we also took advantage of the on node parallelism to compute  computing multiple distances at the same time. 
In each of these experiments, we have decided to apply persistence simplification in order to move our merge trees down to 14 nodes each.

In order to determine a size for the graphs which provided us with a good balance of computational feasibility and low persistence thresholding, we evaluated the computation times of 8,10,12,14,16 and 18 node graphs and the corresponding persistence thresholds which achieved these sizes. We randomly pulled 175 pairs of scalar fields from our shape comparison experiment (\autoref{subsec:shapecompare}) and computed the distance between pairs to track the average computation time. \autoref{fig:computationTime} shows these values as line charts. We can see that using 16 and 18 node graphs begins to increase the computation time dramatically while the persistence simplification needed to attain 16 and 18 node graphs is only slightly less than the persistence simplification value needed to attain 14 node graphs. For reference, the average difference between the global min and global max for this subset of the data was 272.79.
\begin{figure}
    \centering
    \includegraphics[width=0.48\textwidth]{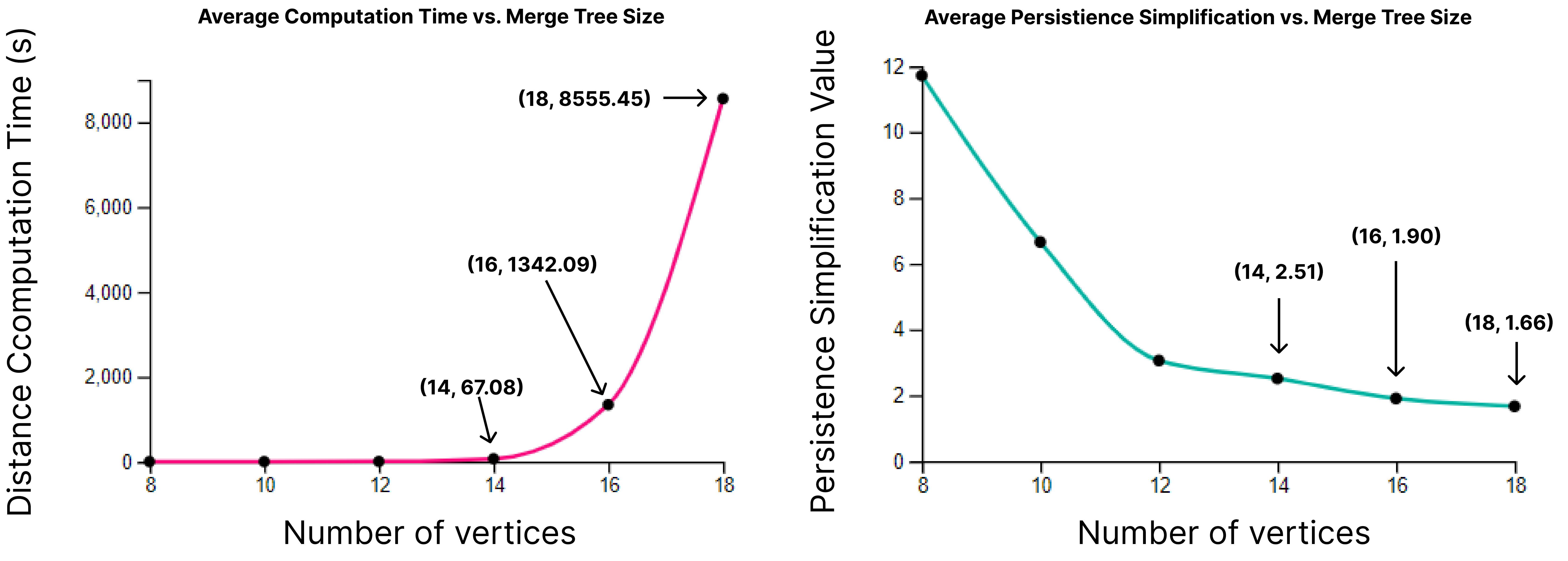}
    \caption{(left) The average computation times for a subsample of the shape comparison dataset with graph sizes 8, 10, 12, 14, 16, and 18. (right) The corresponding average persistence simplification values for this subsample of data.}
    \label{fig:computationTime}
\end{figure}

\subsection{Shape Comparison\label{subsec:shapecompare}}
\begin{figure}[h]
    \centering
    \includegraphics[width=0.45\textwidth]{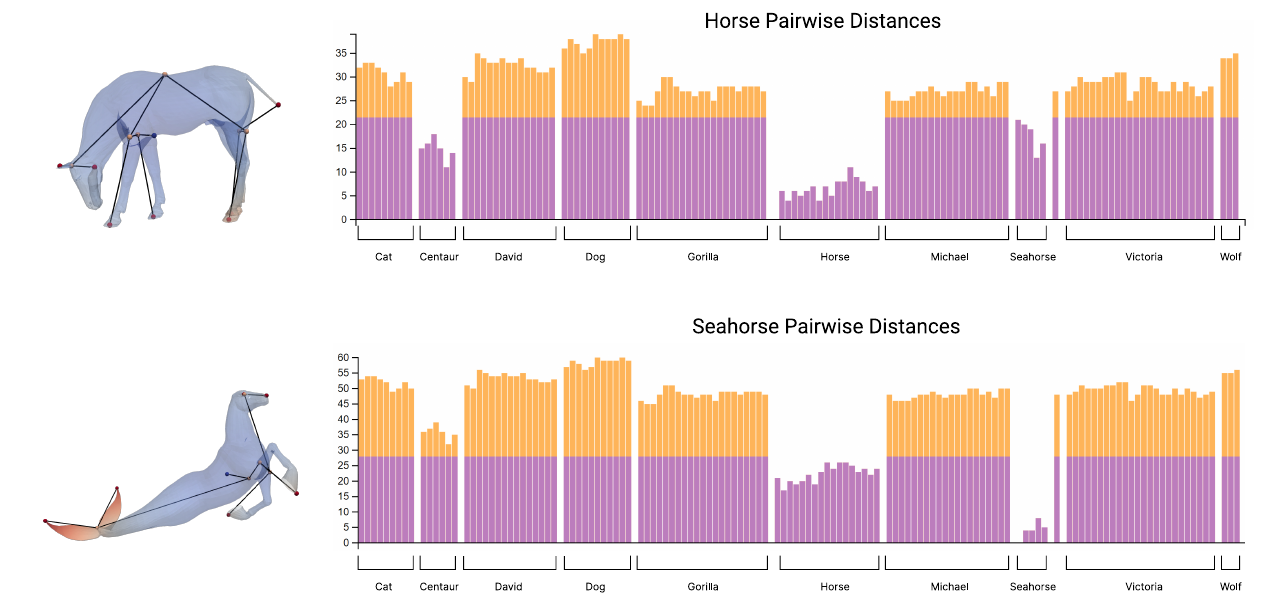}
    \caption{Comparison of a single horse pose and single seahorse pose to the rest of the shape dataset. We note that the horse pairs produce the same distance as the bottleneck distance. This is possibly due to few, if any, topological changes to get from one pose of the horse to another pose. Any of the topological changes must be outweighed non-movement relabels, insertions, or deletions.}
    \label{fig:shapeDiscriminativity}
\end{figure}
We took the TOSCA non-rigid world dataset, which contains a collection of shapes of animals and humans in different poses, and then we computed the average geodesic distance on each of them using the method suggested by Hilaga et al.~\cite{Hilaga2001}. We randomly sampled a subset of 100 vertices from each mesh, calculated the geodesic distance from every vertex to the subset and then took the average. Next we persistence simplified the data, by using a custom threshold for each mesh so that we get 14 nodes in the split tree. The largest persistence simplification value was 7.34 which implies that the distances we have computed are at most 3.67 above the true distance.

\subsubsection{Results}
We computed pairwise distance between 132 shapes, separating the 17424 distance computations into 12 batches. Each batch, with on-node parallelization from the 32 cores, took an of average 44.84 minutes.

\autoref{fig:shapesDistMatrix} shows the pairwise distances computations of our merge tree matching distance compared to the \textbf{merge tree edit distance} from Sridharamurthy et al.~\cite{Sridharamurthy2018}. We note that the distance matrix for merge tree edit distance produced here differs from the distance matrix produced in their original paper. This can be due to several reasons: 1) we used more vertices of the scalar field in order to compute the average geodesic distance, 2) the merge tree edit distance does not simplify the resulting scalar field, and 3) a difference in color scale.

We can expect a difference between our distance and the merge tree edit distance due to our distance being more stable as well as merge tree distance summing the values of the feature differences rather than taking the largest feature difference.

We would like to note that our distance and the merge tree edit distance produce similar global patterns. For example, comparisons to the seahorse produce relatively large distances and comparisons between humanoid shapes produce relatively low distances. In our distance matrix, we can see that there is a low distance for comparison between two of the same classes of shapes, regardless of the pose it takes. Another interesting point is the relationship between the centaur to other shapes. The shape has a similar distance to each of the other shapes, besides the seahorse and shark, which may be expected due to half of the centaur's shape being similar to each of the other shapes.

In \autoref{fig:shapeDiscriminativity}, we depict the bottleneck distance compared to our merge tree matching distance for a single horse and single seahorse to all other poses. We note that the only time that our distance achieves the bottleneck distance exactly is when we compare the horse to other horse poses, centaur poses, or seahorse. This is likely due to the large features being able to be matched to one another with little need for adjacency changes. The case is similar for the seahorse pose.
\label{sec:appx:shapeData}

\subsection{von Kármán Vortex Street}

We obtained a von Kármán Vortex Street dataset from \cite{Guenther17} and calculated the vorticity scalar field for a set of uniformly sampled timesteps. Then we persistence simplified each timestep so that the resulting merge would have exactly 14 nodes. Since we were only planning on obtaining 14 nodes, we decided to clip the vortex data so that more information from the clipped timestep could be obtained. Then we computed the pairwise distance using our distance between time-steps. The distance matrix is shown in \autoref{fig:vortexDistMatrix}. The largest persistence simplification value was 6.60, which implies that the distances we have computed are 3.30 above the true distance.

\subsubsection{Results}

As shown in \autoref{fig:vortexDistMatrix} (right) we can see that there is an initial time period where there are no vortices in the data. This is reflected in the upper left hand cluster of the distance matrix in \autoref{fig:vortexDistMatrix} (left) which shows low distance values between early consecutive timesteps. Once the vortices have started forming, we can see that their positions periodically alternate as the vortices move forward. The collection of alternating high and low values in the bottom right section of the distance matrix demonstrates this periodicity.

\subsection{Stability Testing}
To test the stability of our distance, we constructed a baseline scalar field $(\X_0,f_0)$ with three maxima, two saddles, and the global minimum. The two connected saddles are within $\e$ of each other -- making the scalar field horizontally $\e$-unstable.  We generated 36 new scalar fields by applying a random plane multiplied by a Gaussian, to simulate random noise. The bottleneck, merge tree matching, and $L^{\infty}$ distance was computed for each scalar field when compared to the baseline and plotted in \autoref{fig:teaser}. We found that, as desired, our distance lay between the bottleneck and $L^{\infty}$ distance. 

The perturbation was chosen with several specifics in mind in order to correctly mimic the case of horizontal stability. Let $\e = |\tilde{f}(x_2)-\tilde{f}(x_1)|$. The extrema $x_3,x_4,x_5$ are assigned function values such that each are more than $2\e$ greater than their connected saddle. Otherwise, the perturbation which we apply would create a new scalar field $(\X,\tilde{g})$ such that deletion of $e(x_2,x_4)$, inserting $e(y_2,x_4)$, and adjusting the function value of $x_1$ to be the function value of $y_1$ would be optimal. Furthermore, since we wanted to focus on horizontal instabilities in this experiment, we made sure that the difference between $x_5$ and $x_3$ is larger than $\e$. Otherwise, the perturbation would essentially just be ``reflecting'' the merge tree, i.e. mapping $x_3$ to $y_5$ and $x_5$ to $y_3$.

It is worth noting that if there was no topology change in the data and our perturbation was only causing differences in function value, then our merge tree matching distance would be equal to the bottleneck distance. Thus, this distance is sufficiently capturing a perturbation which changes the topology of the merge tree.

\begin{figure}[h]
    \centering
    \includegraphics[width=0.45\textwidth]{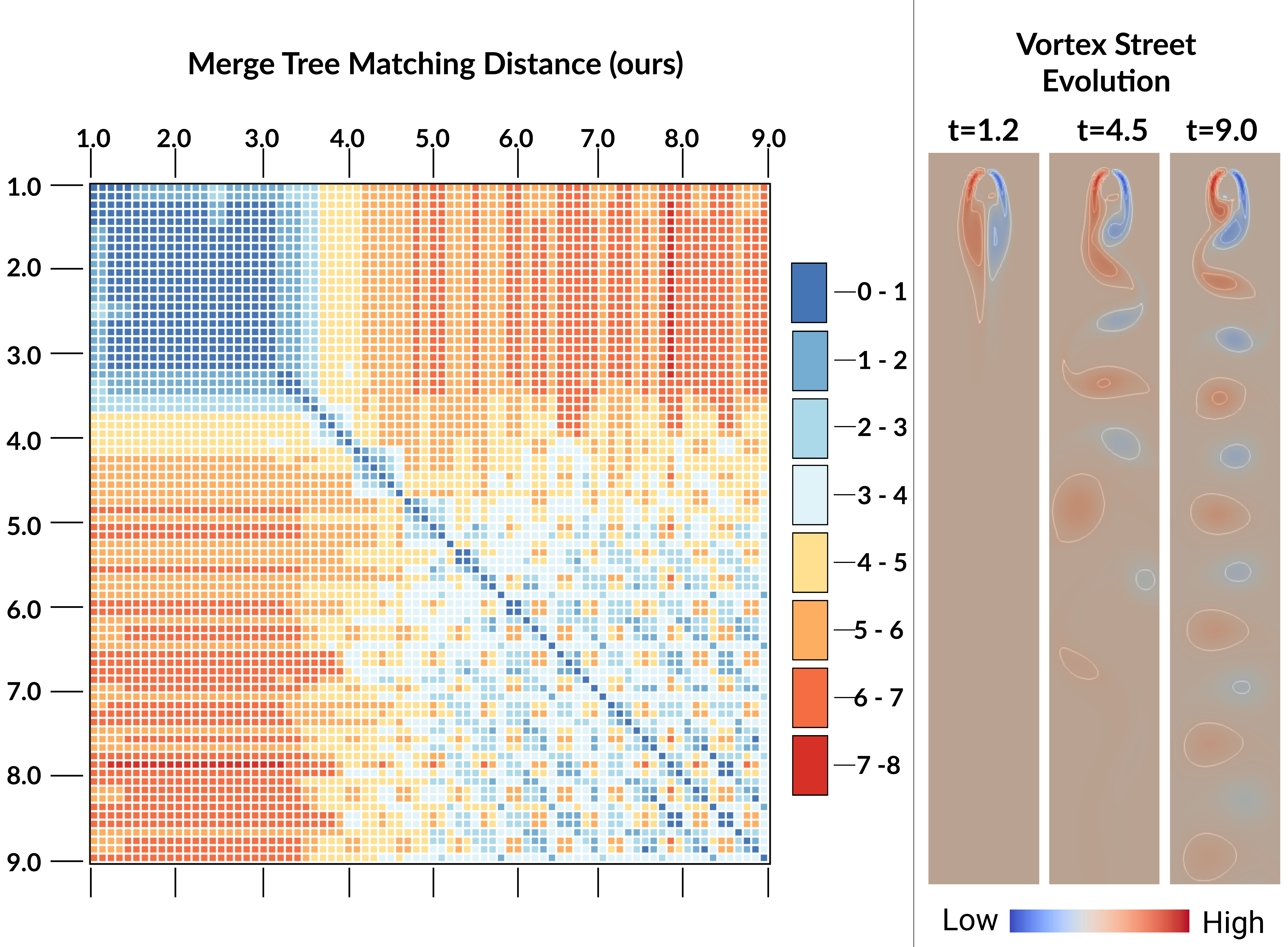}
    \caption{(Left)Pairwise distances between entries in 2-dimensional von Kármán vortex street. (Right) Three different timesteps of the von Kármán vortex street which depict the evolution of the vortices.}
    \label{fig:vortexDistMatrix}
\end{figure}

\section{Discussion}
\label{sec:discussion}

Here we have constructed a distance on merge trees which has experimentally been shown to be both stable and more discriminative than the bottleneck distance. Not only was our distance less than the $L^{\infty}$ distance during our experimentation, but it was specifically designed to identify and quantify scenarios where perturbations in the dataset may cause topological changes in the merge tree or pairing changes in the persistence diagram and branch decomposition trees.

\paragraph{Properties of a Metric} 
While we do not explicitly prove the triangle inequality and symmetry property, we would like to note that we verified that these properties hold on each of the datasets that we provided here.

\paragraph{Comparison to Beketayev Distance}
Beketayev et al. introduced a distance on merge trees which also computes and compares all branch decomposition trees to one another~\cite{Beketayev2014}. They are able to reduce the computation time of comparing all BDTs by not repeating comparisons of subtrees of specific BDTs. One fundamental difference that makes this possible is that once a node $x$ is paired to a node $y$, the children of $x$ must be mapped to the children of $y$, or be inserted/deleted. In our case, we cannot necessarily re-utilize comparisons of BDT subtrees since our nodes are always able to be mapped outside of any given subtree. This particular restriction of ancestor-descendant relationships is exactly what may cause horizontal instabilities while reducing the computation time. It is for a similar reason that a direct application of tree edit distance is unstable on merge trees.

\paragraph{Translation to Contour Trees and Reeb Graphs} 
When translating to contour trees and Reeb graphs, the A* algorithm would still be able to adequately match features of one graph-based descriptor to another. Furthermore, we can introduce additional pruning since these theoretical graph-based distances always have the restriction that we cannot match features of different types to one another (e.g. an up-leaf cannot be matched to a down-leaf in a contour tree or Reeb graph). The hurdle that we run into is that of properly encoding the features of these descriptors. There is some nuance on choosing the pairing in contour trees. For example, the path from the global min to the global max may be a non-monotone path. To the best of our knowledge, there has been no generalization of the BDT for Reeb graphs.

\paragraph{Scalability} We would like to make note that our algorithm still suffers from the issue of scalability. While a strength of our approach is that we can use persistence simplification to reduce the number of vertices while retaining accuracy, even small increases in the size of the merge tree may cause our computation time to increase in an exponential fashion (see \autoref{fig:computationTime}). Nevertheless, while analysis of small trees may be practical in some settings, we see developing a more efficient approach as an important, open challenge.

\acknowledgments{We thank Raghavendra Sridharamurthy and Vijay Natarajan for providing the comparison results of their algorithm~\cite{Sridharamurthy2018} on experiment used in \autoref{subsec:shapecompare}. We also thank our anonymous reviewers for provided their detailed feedback and suggestions. This work is supported in part by the U.S. Department of Energy, Office of Science, Office of Advanced Scientific Computing Research, under Award Number(s) DE-SC-0019039.}

\bibliographystyle{abbrv-doi}

\bibliography{bibFile}

\appendix 

\section{Appendix}
\label{sec:appx:algorithmDetails}

\subsection{A* Algorithm Heuristics and Pruning}
Below we have one pruning tactic and one heuristic function which we implemented into our algorithm. 
\paragraph{Checking Relabel Range and Validity:} When exploring all possible matchings, it is important to remove any possible matchings that lead to suboptimal results. We have two criteria which help prune the possible matches: 1) checking if $u$ and $v$ are close enough in function value so that relabeling them to one another is not more costly than deleting $u$ and inserting $v$ and 2) checking if there exists an ancestor of $u$ that will be a descendant of $v$. 

\begin{definition}
Let $u = (u_s,u_e) \in b_f$ and $v = (v_s,v_e) \in b_g$. Then, let $\delta = \frac12|u_e-u_s|$. We say that $v$ is in the \textbf{relabel range} of $u$ if $u_s - \delta \leq v_s \leq u_s + \delta$ and $u_e - \delta \leq v_e \leq u_e + \delta$.
\end{definition}

If $v$ is not in the relabel range of $u$ and $u$ is not in the relabel range of $v$, then the cost of $(u,v)$ is always greater than the cost of $\max\{c(u,\lambda),c(\lambda,v)\}$. Thus, $u$ should not be mapped to $v$ since this will always lead to a sub-optimal edit sequence.

Let $(u,v),(u',v') \in M$ such that $u$ is the parent of $u'$ and $v$ is a child of $v'$. This means that $u_s < u'_s$ and $v'_s < v_s$, implying that $|f(u_s) - g(v_s)| > |f(u'_s) - g(v_s)|,|f(u_s)-g(v'_s)|$. Thus, the range that the saddles traverse will always be greater than if we chose $(u,v')$ and $(u',v)$ as our pairs instead. Furthermore, since we iterate over all branch decomposition trees, we are guaranteed that each $u_s,u'_s,v_s,v'_s$ are paired with extrema which coincide with minimizing this cost. Thus, if $(u,v) \in M$ with $u'$ being an ancestor of $u$ and $v'$ being a descendant of $v$, we do not allow $(u',v')$ in our matching.

\paragraph{Size difference heuristic}
Let $M'$ be a partial matching between branch decomposition trees $b_f$ and $b_g$. Suppose $U$ and $V$ are the unmatched nodes of $b_f$ and $b_g$. Without loss of generality, suppose $n = |U| - |V| > 0$. This means that in the matching, at least $n$ nodes must be deleted from $b_f$. We can lower bound the actual cost by computing the cost of deleting the $n^{th}$ smallest node from $b_f$. Since this is a lower bound to the true cost of the full matching, we are still guaranteed that this heuristic will be viable for the A* algorithm to reach the optimal value. We use this function as our heuristic function $h(M')$ in the A* algorithm.

\subsection{Zigzag Diagram Example}
\label{sec:appx:zigzagExample}

Figure ~\ref{fig:zigzag} depicts an example of a zigzag diagram. We encode function value using height for the merge trees. Since there is no function values associated on the connecting spaces, the vertical position of the nodes of the connecting spaces do not encode function value unlike the merge trees above them. Each value $x_i$ maps to $y_i$ under both quotient maps $q_{i,i}:X_i \to \MM_i$ and $q_{i,i+1}:X_i\to\MM_{i+1}$. Color in the connecting spaces indicate points which belong to the same sequence. For example, $\{x_3,x_3,x_3,x_3,x_3\},\{x_3,x_4,x_4,x_4,x_4\},\{x_1,x_1,x_1,x_1,x_1\},\{x_1,x_1,x_1,x_1,x_8\}$ are all valid sequences. The spread of each of these sequences is then the range of the associated function values in the sequence of merge trees. For example, $\text{spread}(\{x_3,x_3,x_3,x_3,x_3\}) = |f_1(y_3) - f_3(y_3)|$.

Note that each of the connecting spaces have the edge $e(x_8,x_9)$, which only appears in $\MM_4,\MM_5,\MM_6$. This represents a leaf which is inserted on $\MM_4$. In the previous merge trees, this edge is contracted to a single point and assigned the function value of half its length. This half is chosen to optimize the distance that both its extrema and minima travel.

\begin{figure*}
    \centering
    \includegraphics[width=0.95\textwidth]{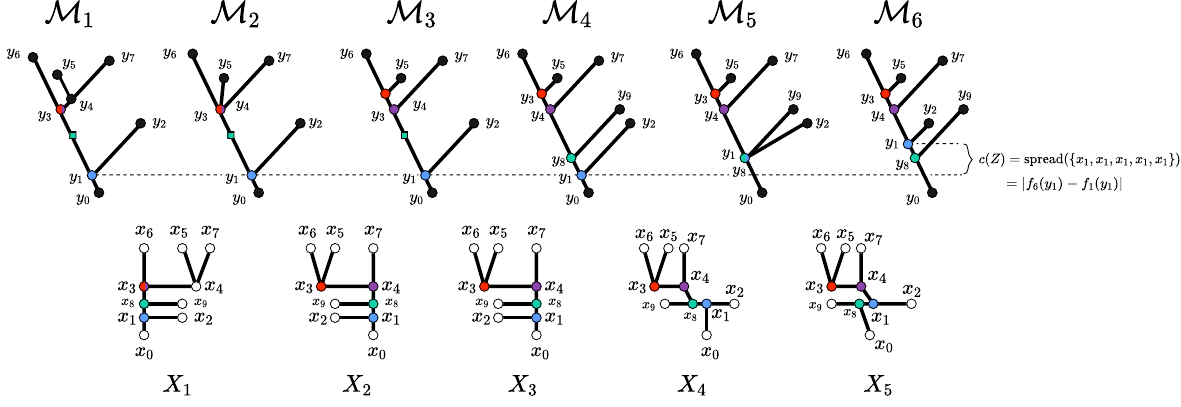}
    \caption{Zigzag diagram carrying a source merge tree $\MM_1$ to a target merge tree $\MM_6$. The connecting spaces shown below, between the two merge trees in which their quotient maps map to. Color indicates that these points all belong to the same sequence. The cost of this zigzag diagram $Z$ is the largest spread over all possible sequences, which is attained by the sequence $\{x_1,x_1,x_1,x_1,x_1\}$.}
    \label{fig:zigzag}
\end{figure*}

\subsection{A* Computation}
Below is pseudocode for the A* computation. Note that the priority queue $Q$ is ordered based on the maximum value between the current cost and the heuristic function, but the current cost is still maintained as a separate value. 
\begin{algorithm}
\caption{aStar}\label{alg:aStar}
\hspace*{\algorithmicindent} \textbf{Input:} Two BDTs $b_f,b_g$, and cuttoff value $\e$ \\
\hspace*{\algorithmicindent} \textbf{Output:} cost of best matching between $b_f,b_g$
\begin{algorithmic}[1]
\State $Q$ = empty priority queue, $U = V(b_f)$,$V = V(b_g)$, $M$ = $\{(r1$,$r2)\}$
\State $Q$.push$\big((\max\{c(M'),h(M')\},c(M),M)\big)$
\While{$|Q| > 0$}
    \State $approxCost$, $c(M)$, $M$ = $Q$.pop()
    \If {$c(M) > \e$}
        \State \Return 0, \texttt{false}
    \EndIf
    \If{$|U| > 0$} \Comment{Pick a $u$ and match to all possible $v$ and $\lambda$}
        \State $u$ = $U$.pop()
        \State $M' = \{M\cup(u,\lambda)\}$
        \State $Q$.push$\big((\max\{c(M'),h(M')\},c(M'),M')\big)$
        \For{$v \in V$}
            \State $M' = \{M\cup(u,v)\}$
            \State $Q$.push$\big((\max\{c(M'),h(M')\},c(M'),M')\big)$
        \EndFor
    \ElsIf{$|V| > 0$} \Comment{Match leftover elements from $v$ to $\lambda$}
        \For{$v \in V$}
            \State $M' = \{M\cup(\lambda,v)\}$
            \State $Q$.push$\big((\max\{c(M'),h(M')\},c(M'),M')\big)$
        \EndFor
    \Else
        \State $c(M)$ = computeFinalCost($M$) \Comment{See Sec. 5.3}
        \If{$c(M)\leq Q$[0]}
            \If {$c(M) \leq \e$}
                \State \Return $c(M)$,\texttt{true}
            \Else
                \State \Return 0, \texttt{false}
            \EndIf
        \Else
            \State $Q$.push$\big((\max\{c(M'),h(M')\},c(M),M)\big)$
        \EndIf
    \EndIf
\EndWhile
\end{algorithmic}
\end{algorithm}

\end{document}